\documentclass{comsoc2020}

\usepackage[dvipdfmx]{graphicx}
\usepackage[dvipdfmx]{graphicx}
\usepackage{amsthm}
\usepackage{multirow}
\usepackage{algorithm}
\usepackage{amsmath}
\usepackage{algorithmic}
\usepackage{hyperref}
\usepackage{amssymb}

\newtheorem{theorem}{Theorem}

\newtheorem{example}{Example}
\newtheorem{lemma}{Lemma}

\newtheorem{remark}{Remark}

\newtheorem{corollary}{Corollary}
\newcommand{\RN}[1]{%
  \textup{\uppercase\expandafter{\romannumeral#1}}%
}
\theoremstyle{definition}

\usepackage{booktabs}
\usepackage{tikz}
\usepackage{graphicx}

\usepackage{mathrsfs}
\usepackage{bm}
\usepackage[margin=1.3in]{geometry}
\usepackage{comment}
\usepackage{subcaption}
\usepackage{tikz}

\newcommand{\W}{\mathcal{W}}
\newcommand{\Bi}{\mathcal{B}}
\newcommand{\PoA}{\mathsf{PoA}}

\renewcommand{\d}{\ensuremath{\mathbf{d}}}
\newcommand{\D}{\ensuremath{\mathbf{D}}}
\renewcommand{\u}{\ensuremath{\mathbf{u}}}

\newcommand{\yuzhe}[1]{\textcolor{orange}{Yuzhe says: #1}}
\newcommand{\davide}[1]{\textcolor{purple}{Davide says: #1}}

\DeclareMathOperator*{\argmax}{arg\,max}

\title{
Tracking Truth by Weighting Proxies \\ in Liquid Democracy
}
\author{Yuzhe Zhang and Davide Grossi}

\begin{document}
\pagestyle{plain}


\begin{abstract}
We study wisdom of the crowd effects in liquid democracy when agents are allowed to apportion weights to proxies by mixing their delegations. We show that in this setting---unlike in the standard one where votes are always delegated in full to one proxy---it becomes possible to identify delegation structures that optimize the truth-tracking accuracy of the group. We contrast this centralized solution with the group accuracy obtained in equilibirium when agents interact by greedily trying to maximize their own individual accuracy through mixed delegations, and study the price of anarchy of these games. While equilibria with mixed delegations may be as bad as in the standard delegations setting, they are never worse and may sometimes be better. 
\end{abstract}


\section{Introduction}

Liquid  democracy \cite{blum2016liquid} is a form of proxy voting \cite{miller1969program,Tullock_1992,Alger_2006,green2015direct,cohensius2017proxy} where each proxy is delegable, thereby giving rise to so-called transitive delegations. In such a system each voter may choose to cast her vote directly, or to delegate her vote to a proxy, who may in turn decide whether to vote or delegate, and so pass the votes she has accrued further to yet another proxy. The voters who decide to retain their votes---the so-called gurus---cast their ballots, which now carry the weight given by the number of delegations they accrued. The system has been deployed in online democratic decision-support tools like LiquidFeedback \cite{liquid_feedback}\footnote{\url{https://liquidfeedback.org/}} and has been object in recent years of research in political science and eDemocracy, as well as artificial intelligence (see \cite{paulin20overview} for a recent overview).

\paragraph{Contribution}
Our paper studies aspects of the truth-tracking properties of liquid democracy when agents are allowed to express delegations consisting of the apportionment of shares of a unit weight (i.e., the agent's voting weight) to their proxies. We interpret these weights probabilistically, that is, as mixing of pure delegations. The issue we are after is to understand the extent to which weighted delegations could help the truth-tracking behavior of liquid democracy. We make three main contributions. {\em First}, we show that in this more general setting it is always possible for the agents to achieve maximal group accuracy by centrally coordinating their delegations (Theorem \ref{thm:algo}).
{\em Second}, we extend the strategic model of liquid democracy developed in \cite{bloembergen19rational,zhang21power} to the setting involving weighted delegations, i.e., mixed delegation strategies. In games in which agents greedily try to maximize their individual accuracies we show that weighted delegations enable equilibria that are better in terms of group accuracy, with respect to equilibria with pure delegations (Theorem \ref{thm:maxNE}). This, however, comes at the cost of a higher price of anarchy with respect to games with pure delegations.
{\em Third}, we provide an interpretation of weighted delegations alternative to mixing, in which weights are modeled as shares of voting power that agents apportion to their proxies. This model leads to an alternative notion of utility in delegation games, and therefore to different equilibria. We prove the resulting notion of equilibrium to be weaker than the mixed delegations one (Theorem \ref{thm:bridge}). 

\smallskip

\paragraph{Related Work}
In the last couple of years, several papers in the area of computational social choice have focused on liquid democracy. Three lines of research have broadly been pursued. 
First, papers have pointed to potential weakenesses of voting by liquid democracy, e.g.: delegation cycles and the failure of individual rationality in multi-issue voting \cite{christoff17binary,brill2018pairwise}; poor accuracy of group decisions as compared to those achievable via direct voting in non-strategic settings \cite{kahng18liquid,caragiannis19contribution}, as well as strategic ones \cite{bloembergen19rational}; issues related to power \cite{zhang21power}. 
Second, a number of papers have focused on the development of better behaved delegation schemes, e.g.: delegations with preferences over trustees \cite{brill2018pairwise} or over gurus \cite{escoffier19convergence,escoffier20iterative}; multiple delegations \cite{golz2018fluid}; complex delegations like delegations to a majority of trustees \cite{colley20smart}; dampened delegations \cite{boldi2011viscous}; breadth-first delegations \cite{Grammateia2020aamas}. 
Third, papers have focused on computational aspects of some of the themes mentioned above, like the computation of equilibria in delegation games \cite{escoffier20iterative}.

Our paper is a contribution to the first line of research mentioned above and is most directly related to \cite{kahng18liquid,caragiannis19contribution}, which studied the truth-tracking properties of liquid democracy as opposed to direct voting. In particular \cite{kahng18liquid} showed that no `local' probabilistic procedure to select proxies can guarantee that liquid democracy is, at the same time, never less accurate (in large enough graphs) and sometimes strictly more accurate  than direct voting. These negative results have been further strengthened along several lines in \cite{caragiannis19contribution}. Like these papers we study delegations from a probabilistic perspective, but we use probabilistic delegations to leverage existing generalizations of the Condorcet jury theorem (c.f., \cite{nitzan82optimal,shapley84optimizing}) in which agents are assigned weights depending on their accuracy. Our starting point is to leverage such results to determine delegation graphs which are optimal in terms of truth-tracking behavior, and which are not feasible in the standard delegation setting. We then compare this centralized approach with a decentralized one in which agents choose their delegations with the sole aim of improving their individual accuracy.

\paragraph{Outline}
Section \ref{sec:preliminaries} introduces our model of liquid democracy and describes a centralized mechanism to use weighted delegations in order to achieve optimal group accuracy. Sections \ref{sec:games} and \ref{sec:welfare} studies weighted delegations from a game-theoretic perspective and focuses on the quality of equilibria in those games measured as group accuracy. Finally, Section \ref{sec:shares} studies an alternative model of weighted delegations based on the transfer of voting power. Section \ref{sec:conclusions} concludes.


\section{Preliminaries} \label{sec:preliminaries}

\subsection{Basic Model: Binary Voting for Truth-Tracking}

Our model is based on the binary voting setting for truth-tracking \cite{Condorcet1785,grofman83thirteen,elkind16rationalizations}. The setting has already been applied to the study of liquid democracy by \cite{kahng18liquid,bloembergen19rational,caragiannis19contribution}.

\smallskip

A finite set of agents $N=\{1,2,\dots,n\}$ has to vote on whether to accept or reject an issue. The vote is supposed to track the correct state of the world---that is whether it is `best' to accept or reject the issue. The agents' ability to make the right choice (i.e., the agents' error model) is represented by the agent's \emph{accuracy} $q_i \in [0.5,1]$, for $i\in N$. Each agent is endowed with one vote and the result of such an election is determined by simple majority.

When agent $i\in N$ delegates to agent $j \in N$ we write $d_i = j$. Then $\d=(d_1,d_2,\dots, d_n)$ is called a \emph{delegation profile} (or simply \emph{profile}) and is a vector describing each agent's delegation. Equivalently, delegation profiles can be usefully thought of as maps $\d: N \rightarrow N$, where $\d(i)=d_i$.
When $d_i=i$, agent $i$ votes on her own behalf. We call such an agent a \emph{guru}. On the other hand, any agent who is not a guru, is called a {\em delegator}. For profile $\d$, and $C\subseteq N$, let $C^{\d}$ denote all gurus in $C$ in profile $\d$, i.e., $C^{\d}=\{i \in C\mid d_i=i \}$. A delegation profile in which all agents are gurus (i.e., for all $i \in N$ $d_i = i$) is said to be {\em trivial}.

Any profile $\d$ can also be represented by a directed graph. An edge from agent $i$ to $j$ exists whenever $d_i=j$.
Consider then a profile $\d$ where a path exists from $i$ to $j$, i.e., $i\rightarrow i_1 \rightarrow \dots \rightarrow i_k \rightarrow j$.
We call such paths \emph{delegation chains}. When such a chain from $i$ to guru $j$ exists, every agent on this delegation chain (indirectly) delegates to $j$, and we denote $i$'s guru by $d^*_i = \d^*(i) = j$. 
For any profile $\d$ and any guru $i\in N^\d$, let $w_\d(i)$ denote the weight collected by $i$, i.e., $w_\d(i)=|\{j\in N\mid \d^*(j)=i\}|$.
A {\em delegation cycle} is a chain where the first and last agents coincide. In such a case, no agent in the chain is linked to a guru. Therefore no agent linked via a delegation chain to an agent in a delegation cycle has a guru.

In several models for liquid democracy agents are thought of as nodes in an underlying network (e.g., \cite{kahng18liquid,bloembergen19rational}). An underlying network is a directed graph $R=\langle N, E\rangle$, where $E$ are directed edges.
Let $R(i)$ denote the {\em neighborhood} of agent $i\in N$, i.e., $R(i)=\{i\}\cup \{j\in N\mid (i,j)\in E\}$.
In those models, agents are only allowed to delegate agents in their neighborhoods.
In this paper, unless clearly specified, we focus on the setting in which such a network is complete, that is, every agent can delegate to every agent.

\subsection{Weighted Delegations}\label{sec:splitvspure}

We are interested in the setting in which, instead of transferring their full vote to a proxy, agents apportion weights to different proxies where then the weight attached to a proxy corresponds to the probability of delegating to it (cf. the notion of local delegation mechanism in \cite{kahng18liquid}). So $i$'s mixed delegation amounts to a stochastic vector $\D_i=(D_{i1},\dots, D_{in})\in \mathbb{R}^n$ with $\sum_{j \in N} D_{ij} = 1$. We call such delegations weight-apportioning delegations or, simply, {\em weighted delegations}. A profile of weighted delegations (or, {\em weighted profile}) is an $n\times n$-dimensional stochastic matrix $\D=(\D_1,\dots , \D_n)$, and $\mathbb{D}$ is the collection of all weighted profiles.
Moreover we can generate a directed graph, which is called a delegation graph, given a weighted profile $\D$.
All agents are nodes in the delegation graph, and for any pair of $i,j\in N$, an directed edge from $i$ to $j$ exists if $D_{ij}>0$.
Clearly, a standard delegation $d_i = j$ corresponds to the vector in which the whole of $i$'s weight is apportioned to $j$, and so the $j^\mathit{th}$ coordinate has value $1$ and the remaining coordinates have $0$. We will be referring to the standard type of delegations also as {\em pure delegations}.

\paragraph{Agents' weights}
Taking a weighted profile $\D$ to denote the probability of all pure delegation strategies of each agent, we can obtain the probability of any standard profile $\d$ as $\Pr (\d)=\Pi_{i\in N}D_{i\d(i)}$.
Then, given a weighted profile $\D$, for any agent $i\in N$, let $\mathcal{D}_i$ denote all possible profiles, in which $i$ is a guru, i.e., $\mathcal{D}_i=\{\d\mid \forall j\in N, D_{j\d(j)}>0, i\in N^\d\}$.
Hence we denote the {\em weight} of any agent $i\in N$ under $\D$ as $i$'s expected weight in all possible standard profiles, that is: 
\begin{align}
    w_\D(i)=\sum_{\d\in \mathcal{D}_i}\Pr (\d)w_\d(i). \label{eq:weight}
\end{align}
In weighted profiles we call all agents, who have positive weight, gurus and denote their set $N^\D=\{i\in N\mid w_\D(i)>0\}$.

\paragraph{Group accuracy}
We call group accuracy (or majority accuracy) the probability that the group correctly identifies the state of the world via a weighted majority vote. The majority accuracy $q_\D$ of a weighted profile $\D$ is therefore the probability that a randomly sampled coalition of gurus $C \subseteq N^\D$, with property $\sum_{i \in C} w_\D(i)\ge \sum_{i \in N^\D \backslash C} w_\D(i)$, contains agents that all vote correctly, while all agents in $N^\D \backslash C$ vote incorrectly. That is:
\begin{align}
q_\D & =\sum_{C \in \W(\D)}\prod_{i \in C} q_i \prod_{i\in N^\D \backslash C}(1-q_i)
\end{align}
where $\W(\D) = \{ C \subseteq N^\D \mid \sum_{i \in C} w_\D(i) \ge \sum_{i \in N^\D \backslash C}w_\D(i) \}$ is the set of winning coalitions. When dealing with standard delegation profiles $\d$ we will also write $q_\d$ for the majority accuracy determined by $\d$.

\begin{example}
\label{ex:ma}
Consider an example, where a set of agent $N=\{1,2,3,4,5\}$ with accuracies $(0.9,0.9,0.6,0.6,0.6)$. Assume the given weighted profile is $\D$, in which $D_{11}=1$, $D_{22}=1$, $D_{33}=1$, $D_4=(0.4,0.4,0.2,0,0)$, and $D_5=(0.4, 0.3, 0.3, 0, 0)$, that is, agents $1$, $2$ and $3$ always vote on their behalves, and the vectors $D_4$ and $D_5$ show the weighted delegation strategy of agent $4$ and $5$, e.g., agent $4$ delegates to agent $1$ by probability of 40\%, to agent $2$ by 40\%, and to agent $3$ by 20\%. Then under $\D$, we can compute agents' weights by Eq.~\ref{eq:weight} as follows:
$$(w_\D(1),w_\D(2),w_\D(3),w_\D(4),w_\D(5))=(1.8, 1.7, 1.5, 0,0).$$
Therefore the set of gurus are $N^\D=\{1,2,3\}$, and by weighted majority rule, all winning coalitions are $\W(\D)=\{\{1,2\},\{1,3\},\{2,3\},\{1,2,3\}\}$.
Then take an instance of the winning coalition $\{1,2\}$, such that the accuracy of this majority is computed by assuming agent $1$ and $2$ vote correctly, i.e., with probabilities (or their accuracies) $q_1=0.9$ and $q_2=0.9$, while agent $3$ votes incorrectly with probability $1-q_3=0.4$.
That is $q_1q_2(1-q_3)=0.324$.
Hence the majority accuracy is the sum of all accuracies of such majorities, which is as follows,
$$q_\D=\sum_{C\in \W(\D)}\prod_{i\in C}q_i\prod_{i\in N^\D\setminus C}(1-q_i)=0.324+0.054+0.054+0.486=0.918.$$
\end{example}


\subsection{Optimal Truth-Tracking by Weighted Delegations}

Our motivation to study weighted delegations comes from generalizations of the Condorcet jury theorem (e.g., \cite[Th. 13]{grofman83thirteen}) showing that the chance that the voting outcome of the group is correct is maximized if a weighted majority rule is used, instead of simple majority, with agents' weights proportional to $\log \left( \frac{q_i}{1-q_i} \right)$ (cf., \cite{nitzan82optimal,shapley84optimizing}). 

By exploiting this observation we can show that weighted delegations offer a viable approach to the optimal delegation problem, that is: given a set of agents with individual accuracies, what is the (weighted) delegation graph that maximizes group accuracy? It is possible to come up with an algorithm that uses one-hop mixed delegations to reallocate weight from the less accurate to the more accurate voters in the group. Let us define the optimal weight of each agent $i \in N$ by:
\begin{align}
    w^\star_i & = n \cdot \frac{\log \frac{q_i}{1-q_i}}{\sum_{j\in N}\log \frac{q_j}{1-q_j}}
\end{align}
Notice that this weight is larger than $1$ for the more accurate agents whereas it is smaller for the less accurate ones and, as desired, it is proportional to $\log{\frac{q_i}{1-q_i}}$. The idea behind the algorithm (Algorithm \ref{algo:maxMC}) is then to have the agents $i$ with $w^\star_i > 1$ apportion their full weight to themselves, and have each agent $j$ with $w^\star_j < 1$ apportion share $w^\star_i-1$ of the excess weight $1-w^\star_j$ to each agent $i$, normalized by the total excess weight of the $i$ agents.

\begin{algorithm}[t]
\caption{Optimal weights}
\begin{description}
\item[Initialize:]
$w=0$, $N_1=\{i\in N\mid w^\star_i<1\}$, $N_2=\{i\in N\mid w^\star_i>1\}$.
\item[Delegate:]\hfill 
\begin{algorithmic}
\FOR{$i\in N_2$}
\STATE{$D_{i,i}=1$}
\STATE{$w=w+(w^\star_i-1)$}
\ENDFOR
\FOR{$i\in N_1$}
\FOR{$j\in N_2$}
\STATE{$D_{i,j}=(1-w^\star_i)\frac{w^\star_j-1}{w}$}
\ENDFOR
\STATE{$D_{i,i}=1-\sum_{k\in N_2}D_{i,k}$}
\ENDFOR
\end{algorithmic}
\item[Return:]$\D$
\end{description}
\label{algo:maxMC}
\end{algorithm}

\begin{theorem} \label{thm:algo}
Algorithm~\ref{algo:maxMC} outputs an element of: $$\argmax_{\D \in \mathbb{D}} q_\D$$
\end{theorem}
\begin{proof}
First we show that in the weighted profile $\D$ returned by Algorithm~\ref{algo:maxMC}, the weight of any agent $i\in N$ is $w^\star_i$.
It can be observed that for any agent $i\in N$, with $w^\star_i<1$, $D_{ii}=w^\star_i$, $D_{ij}=0$ for all $j\in N_1\setminus \{i\}$, and $D_{ij'}=(1-w^\star_i)\frac{w^\star_{j'}-1}{w}$ for all $j'\in N_2$.
On the other hand, for all $i\in N_2$, i.e., agents with $w^\star_i>1$, $D_{ii}=1$.
Then for any $i\in N_1$, the probability of any $\d\in \mathcal{D}_i$, i.e., $i$ is a guru in $\d$, or equivalently $i$ delegates to herself, is $D_{ii}$.
Moreover, for all $\d\in \mathcal{D}_i$, $w_\d(i)=1$.
Thus we have for all $i\in N_1$, $$w_\D(i)=\sum_{d\in \mathcal{D}_i}\Pr(\d)w_\d(i)=D_{ii}=w^\star_i.$$
Then for any $i\in N_2$, $i$ is always a guru since $D_{ii}=1$, and let $\mathcal{D}_{i,ji}$ denote the set of profiles $\mathcal{D}_{i,ji}=\{\d\mid i\in N^\d, \d(j)=i,\forall j'\in N, D_{j'\d(j')>0}\}$ for any $j\in N_1$.
That is an agent $j\in N_1$ delegates to $i\in N_2$ in any profile in $\mathcal{D}_{i,ji}$.
Then the probability of any profile in $\mathcal{D}_{i,ji}$ is $(1-w^\star_j)\frac{w^\star_i-1}{w}$, which is the weight-apportioning delegation from $j$ to $i$.
Also note that $i$ always retains her weight of 1 in any $\d\in \mathcal{D}_i$.
Therefore we have that
$$w_\D(i)=\sum_{\d\in \mathcal{D}_i}\Pr(\d)w_\d(i)=1+\sum_{j\in N_1}(1-w^\star_j)\frac{w^\star_i-1}{w}.$$
Notice that all delegations received by $N_2$ should equal the amount from $N_1$, i.e., $\sum_{j\in N_1}(1-w^\star_j)=w$.
Hence $w_\D(i)=w^\star_i$ for all $i\in N_2$ and therefore, by \cite[Th. 13]{grofman83thirteen}, $q_\D$ is maximal.
\end{proof}

\begin{example}
\label{ex:maxMA}
Let us again use Example~\ref{ex:ma} to show the dynamics of Algorithm~\ref{algo:maxMC}.
We first compute $\log \frac{q_i}{1-q_i}$ for all $i\in N$ as $(0.9542, 0.9542, 0.1761, 0.1761, 0.1761)$.
Then $w^\star_i$ is computed by normalizing the above vector by the entire weight, which is in total $5$ votes, and for all $i\in N$, $w^\star_i$ are $(1.958, 1.958, 0.3613, 0.3613, 0.3613)$.
Hence we can observe that $w^\star_1$ and $w^\star_2$ are larger than their initial weight $1$, and then they do not delegate.
For any other agents $i\in \{3,4,5\}$, they delegate the surplus weight above $w^\star_i=0.3613$ to agents $1$ and $2$ equally since $w^\star_1=w^\star_2$.
Therefore the returned profile is $\D$, in which $D_{11}=D_{22}=1$, $D_3=(0.31935, 0.31935, 0.3613, 0,0)$, $D_4=(0.31935, 0.31935, 0, 0.3613,0)$, and $D_5=(0.31935, 0.31935, 0, 0,0.3613)$.
Then for all $i\in N$, $w_\D(i)=w^\star_i$.
\end{example}

It is worth observing that since pure delegation profiles are just degenerate weighted delegation profiles (so, they are also elements of $\mathbb{D}$) the optimal group accuracy reachable by weighted delegations cannot be possibly worse than the optimal group accuracy obtainable by pure delegations. But, crucially, it can be better, as shown in the following example.
\begin{example}\label{ex:difference}
Let us continue with Example~\ref{ex:maxMA}. In that example the optimal pure delegation profile is the one in which only one delegation happens: an agent with accuracy of $0.6$ delegates to an agent with accuracy of $0.9$.
Then the optimal (pure profile) majority accuracy is $0.918$, which is lower than the optimal accuracy, $0.92664$, of the weighted profile $\D$ in Example~\ref{ex:maxMA}.
\end{example}
Intuitively, pure delegations allow for only discrete weights and can therefore  only roughly approximate a weight distribution among gurus in which each winning coalition $C$ of agents is more accurate than the corresponding losing coalition $N \backslash C$.

\medskip

Observe finally that Algorithm~\ref{algo:maxMC} is polynomial in $N$. So it provides us with a tractable centralized mechanism to solve the optimal delegation problem when weighted delegations are available via mixing. In the next two sections we move from this centralized perspective to a decentralized one, exploring the behavior of weighted delegations in a game-theoretic setting.

\begin{remark}
It is worth to also briefly discuss Algorithm \ref{algo:maxMC} in the context of the GreedyCap algorithm proposed in \cite{kahng18liquid}. GreedyCap is a local probabilistic delegation algorithm, with a centralized element: a cap on the maximal number of delegations. The cap essentially guarantees that delegations do not `break' the wisdom of the crowd effect of independent voters by introducing too much correlation. Algorithm \ref{algo:maxMC} implements a fully centralized approach to group accuracy by assuming delegations, and their weights, to be centrally determined.
\end{remark}


\section{Decentralized Interaction} \label{sec:games}

Algorithm~\ref{algo:maxMC} described a centralized mechanism to establish optimal weighted delegations within a group involved in a truth-tracking task. We move now to study the delegation structures that would arise if agents were to autonomously decide how to delegate. 

\subsection{Incentives to Delegate}
\label{sec:incentive}

We assume each agent $i$ holds an evaluation function $u_i: N \rightarrow \mathbb{R}$
that associates a cardinal value to any agent $i\in N$. Intuitively $u_i(j)$, with $i, j \in N$, is the utility that $i$ obtains if $j$ acts as $i$'s guru. For example, $u_i$ may be such that $u_i(j) = q_j$, that is, $i$'s utility would amount to the accuracy of $j$ if $j$ were to act as $i$'s guru (this is for instance the case in the game-theoretic setting of \cite{bloembergen19rational}).
Note that the utility given by function $u$ is independent of the specific delegation chains linking agents to their gurus: in two profiles $\d$ and $\d'$, for agent $i\in N$, $\d(i)\not= \d'(i)$ but $\d^*(i)=\d'^*(i)$, then the utility of $i$ is identical in the two profiles.
Slightly abusing notations, let $\u_i$ be the $n$-dimensional vector, of which the $j$-th ($j\in N$) element is $u_i(j)$.

\smallskip

Given a weighted profile $\D$, let $\mathcal{D}$ be the set of possible pure profiles with positive probability, i.e., $\mathcal{D}=\{\d\mid \forall i\in N, D_{i\d(i)}>0\}$.
Then $U_i(\D)$ is the utility of agent $i\in N$ under any weighted profile $\D$.
$U_i(\D)$ is a mapping $U:\mathbb{D}\rightarrow (N\rightarrow \mathbb{R})$, and defined as 
\begin{align}
    U_i(\D) & =\sum_{\d\in \mathcal{D}}u_i(d_i^*)\Pr(\d).\label{eq:u_prob}
\end{align}
That is, $U_i(\D)$ corresponds to $i$'s expectation over the utility it would obtain given the distribution over its possible gurus that is induced by $\D$.

\subsection{Delegation Games}

Equipped with the notion of utility we move to define delegation games as structures $G = \langle N, S, U \rangle$, where $N = \{1,2,\dots, n\}$ is the set of agents, $S = \mathbb{D}$ is the strategy space\footnote{Here we only consider delegation in complete networks. We can restrict the strategy space to $S\subset \mathbb{D}$ for general networks, in which agents are only allowed to delegate to their neighbors in the network.} of each agent $i\in N$ (so a strategy is the mixing of pure delegations), and $U$ is as in equation \eqref{eq:u_prob}. Observe that, since utility function $U$ corresponds to the expectation over $u$ given a profile $\D$, the corresponding delegation game can be viewed as the mixed-strategy version of the delegation game with pure delegations. By Nash theorem (cf. \cite{osborne94course}) we therefore know that such games always have Nash Equilibria (NE), and we call such NE under weighted strategy $U$-NE. That is not the case for the pure delegation variant of the game (cf. \cite{10.1007/978-3-030-30473-7_19}).

Let us also fix some auxiliary notation. Given a standard profile $\d$, a deviation by $i\in N$ from $\d(i)$ to $d'_i$ yields the profile denoted $(\d_{-i},d'_i)$, in which any agent in $N\setminus \{i\}$ delegates in the same way as in $\d$, but $i$ delegates to $d'_i$.
This notion of deviation can be extended to weighted profiles in the natural way where $(\D_{-i},D'_i)$ is the new profile, and $D'_i$ is an $n$-dimensional stochastic vector.

\begin{example}[A $2$-agent delegation game] \label{ex:game}
Consider a game with $N=\{1,2\}$,  $\u(1)=(x_{11},x_{12})$ and $\u(2)=(x_{21},x_{22})$, where $x_{12} > x_{11}$ and $x_{21} > x_{22}$.
That is, agent $1$ prefers delegating to $2$ over being a guru, and agent $2$ prefers delegating to $1$ over being a guru.

If the strategy of agent $1$ is $\D^1_1=(1,0)$, i.e., agent $1$ delegates to herself by probability of $1$, her utility is $$U_1(\D^1)=x_{11},$$
since in any possible profile $\d$, $\d(1)=1$.
However, if the strategy of agent $1$ is $\D^2_1=(0,1)$, her utility becomes 
$$U_1(\D^2)= y_{22} \cdot x_{12},$$
where the strategy of agent $2$ is $(y_{21},y_{22})$.
This is the case because $y_{21}>0$, in the corresponding delegation graph the two agents form a delegation cycle and their utility is therefore $0$ in that case.

Therefore agent $1$ would choose $\D^1_1$ if $y_{22}<x_{11}/x_{22}$, otherwise she chooses $\D^2_1$.
Similarly, assume that the strategy of agent $1$ is $(y_{11},y_{12})$.
For agent $2$, we obtain that agent $2$ chooses $\D^3_2=(0,1)$, i.e., always being a guru, if $y_{11}<x_{22}/x_{21}$, otherwise she chooses $\D^4_2=(1,0)$.
The best strategies of both agents are shown in Fig.~\ref{fig:ex1}, and the $U$-NE are $\begin{pmatrix}
1&0\\
1&0
\end{pmatrix}$,
$\begin{pmatrix}
\frac{x_{22}}{x_{21}}&1-\frac{x_{22}}{x_{21}}\\
1-\frac{x_{11}}{x_{12}}&\frac{x_{11}}{x_{12}}
\end{pmatrix}$, and
$\begin{pmatrix}
0&1\\
0&1
\end{pmatrix}$.

\begin{figure}
\centering
\includegraphics[scale=0.2]{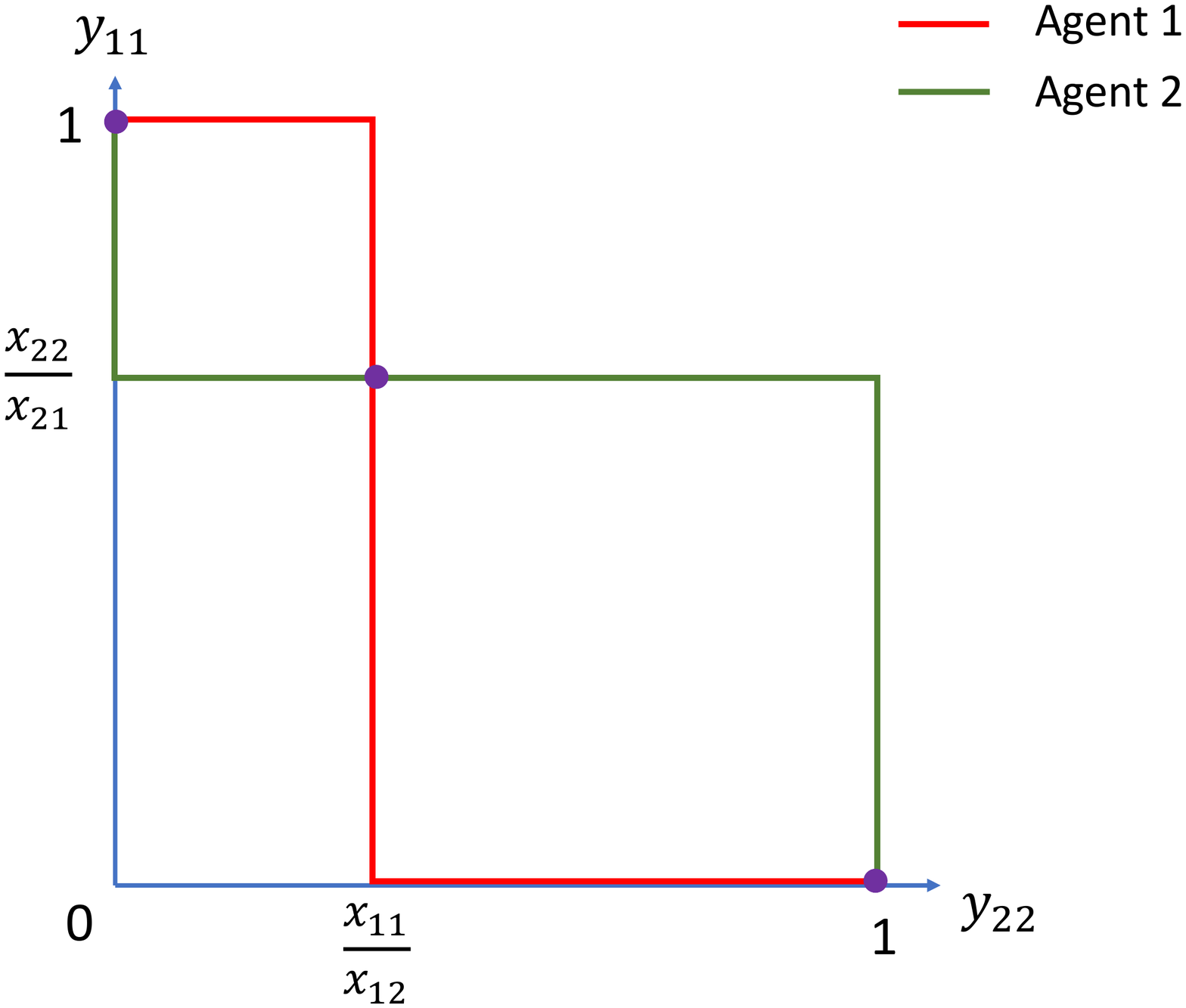}
\caption{Equilibria for the delegation game of Example \ref{ex:game}
}
\label{fig:ex1}
\end{figure}
\end{example}


\subsection{Greedy Delegation (GD) Games}
In the remainder of the paper we focus on a specific class of delegation games in which the utility $u_i$ of agent $i$ is set to equal the accuracy of $i$'s guru. So agents delegate by greedily trying to optimize their own individual accuracy. A {\em greedy delegation game} (GD game) is a delegation game as defined above where, however, the utility function $U$ is such that:
\begin{align}
    U_i(\D) & =\sum_{\d\in \mathcal{D}} q_{\d^*(i)} \Pr(\d).
\end{align}
So, in a GD game all agents evaluate gurus in the same way and $i$'s utility corresponds to the expected accuracy of her gurus. The pure delegation variant of these games has been studied in \cite{bloembergen19rational}, where it has been proven that pure strategy NE always exist for them.

\smallskip

We observe that NE in GD games never contain weighted delegation cycles. This observation generalizes results for pure delegation games from \cite{bloembergen19rational} and it will be of use later in Section \ref{sec:shares}.
In what follows, let $[k]=\{1,2,\dots, k\}$ for any $k\in \mathbb{Z}_{++}$.

\begin{lemma}
\label{lemma:acyclic}
NE in GD games are acyclic.
\end{lemma}

\begin{proof}
Note that this proof is for general underlying networks, and it henceforth holds for complete networks.
We prove by contradiction.
Assume that in a weighted strategy NE $\D$, there is a cycle in the corresponding delegation graph, i.e., there is a sequence of agents $(i_1,i_2,\dots, i_\ell)$ (let $N^\ell$ denote the set of these agents) such that $D_{i_j,i_{j+1}}>0$ for all $j\in [\ell-1]$and $D_{i_\ell,i_1}>0$.
Observe that since $\D$ is a weighted strategy NE, in the delegation graph of $\D$, for any agent $i\in N$, any edge from the agent is on a path to the highest-accuracy agent, who is reachable from $i$ in the underlying network.
Otherwise some agent in the network must be able to change her strategy to delegate to the reachable highest-accuracy agent, since all agents have identical utility function.
Assume that the highest accuracy the set of agents $N^\ell$ can reach is $q^*$, that is in the delegation graph, there is a path from any agent in $N^\ell$ to an agent with accuracy of $q^*$.
Then we consider two possible cases:
(1) An agent $i\in N^\ell$ has accuracy of $q^*$.
Then by adjusting her weighted strategy to $(\D_{-i}, (\dots, D'_{ii}=1, \dots))$, $i$ can strictly improve her utility since she obtains $0$ in the profiles with the delegation cycle by probability of $\prod_{a\in [\ell]}D_{i_a i_{a+1}}D_{i_\ell i_1}$. Contradiction.
(2) All agents in $N^\ell$ has lower accuracy than $q^*$.
Then in $N^\ell$, we can always find an agent, for instance $i_1$ without loss of generality, such that $\exists k\in R(i_1)\setminus N^\ell$, the edge $(i_1,k)$ is on a path in the delegation graph to an agent with accuracy of $q^*$.
By adjusting $i_1$'s weighted strategy to $(\D_{-i_1}, (\dots, D'_{i_1 k}=1, \dots))$, $i_1$ can also strictly improve her utility. This goes against the assumption that $\D$ is a NE. Contradiction.
\end{proof}


\section{Social Welfare in Greedy Delegation Games} \label{sec:welfare}

We measure social welfare in GD games by group accuracy and study this measure of welfare in equilibrium.


\subsection{Group Accuracy in Equilibrium}

Weighted delegations make it possible to achieve wisdom of the crowd effects in equilibrium by balancing weight among maximally accurate agents. As a result NE in GD games with weighted profiles can be shown to be never worse that NE with pure delegations and to be better in some cases.
\begin{theorem}
\label{thm:maxNE}
For any GD game and let $N^*\subseteq N$ the set of agents with maximal accuracy. Then any weighted profile $\D$, such that for any $i\in N^*$, $w_\D(i)=n/m$ ($|N^*|=m$), is the NE with maximal accuracy.
\end{theorem}
\begin{proof}
Assume that the set of agents $N^*$ has identical highest accuracy of $q^*$.
First notice that any weighted profile $\D$, such that $\sum_{i\in N^*}w_\D(i)=n$, is a mixed NE, since any agent in $N$ then has a set of gurus in $N^*$ and acquires utility of $q^*$.
Therefore any weighted profile $\D^*$, such that for any $i\in N^*$, $w_{\D^*}(i)=n/m$, is a mixed NE.

Then in $\D^*$, we consider any bipartition $(N_1,N_2)\in \Bi(N^*)$ (where $\Bi(N^*)$ denotes the set of all bipartitions in $N^*$), where $N_1$ is the majority, i.e., $\sum_{a\in N_1}w_{\D^*}(a) > \sum_{a\in N_2}w_{\D^*}(a)$.
Therefore the accuracy of this bipartition is $\prod_{a\in N_1}q_a\prod_{a\in N_2}(1-q_a)$.
Since $\forall i,j \in N^*$, $q_i=q_j=q^*$ and $w_{\D^*}(i)=w_{\D^*}(j)$, we have $|N_1|>|N_2|$ ,which leads to
$${q^*}^{|N_1|}(1-{q^*}^{|N_2|})=\prod_{a\in N_1}q_a\prod_{a\in N_2}(1-q_a)> \prod_{a\in N_2}q_a\prod_{a\in N_1}(1-q_a)={q^*}^{|N_2|}(1-{q^*}^{|N_1|}).$$ 
Therefore if $N^*$ is the set of gurus, $\D^*$ maximizes the majority accuracy. By by~\cite[Th. 1]{shapley84optimizing} no NE $\D$ in which $N^\D$ is a strict subset of $N^*$ has higher group accuracy.
\end{proof}

Then the following example shows that there exist GD games in which $\D^*$ has strictly better group accuracy than any NE with pure delegations.
\begin{example}
\label{ex:strictbetterma}
Consider a GD game, where there are $7$ agents, $5$ agents with maximal accuracy ($q^*$). In this example, any pure delegations NE with maximal group accuracy, would involve a pair of maximally accurate agents who get each a weight of $2$.
These two agents form a winning coalition, but they have a lower group accuracy than the remaining three gurus, i.e., ${q^*}^2(1-q^*)^3< {q^*}^3(1-q^*)^2$. So the resulting group accuracy is strictly worse than that of $\D^*$.
\end{example}


\subsection{Price of Anarchy}

We define the price of anarchy of GD games as follows:
\begin{align}
    \PoA & = \frac{\max_{\D \in \mathbb{D}} q_\D}{\min_{\D \in \mathit{NE}}q_\D} \label{eq:poa}
\end{align}
where $\mathit{NE}$ denotes the set of equilibria. When restricting to the price of anarchy in games with pure delegations (and therefore one pure delegations NE) we refer to $\PoA^\mathit{pure}$. 

Equation \eqref{eq:poa} gives us a measure of how much group accuracy is 'lost' in equilibrium, in the worst case, with respect to what would be centrally achievable via Algorithm~\ref{algo:maxMC}.

\begin{theorem}
\label{thm:poa1}
When $|N| \to \infty$, $\PoA \to \frac{1}{q^*}$, where $q^*$ is the accuracy of a maximally accurate agent in $N$.
\end{theorem}
\begin{proof}
Let $\D$ be the weighted profile for which $q_\D$ is maximal and let $\dot{\D}$ be the equilibrium profile for which $q_{\dot{\D}}$ is minimal. This is the case when all agents delegate to the same guru, which has accuracy $q^*$ since it is in equilibrium by assumption. Then, by \cite[Th. 13]{grofman83thirteen}, and the Condorcet jury theorem, as $|N|\rightarrow \infty$, $q_\D\rightarrow 1$ and by construction $q_{\dot{\D}} =  q^*$. 
\end{proof}
It is worth noticing that the same argument can be applied to the setting with pure delegations, obtaining the same asymptotic value for the price of anarchy.

\smallskip

Intuitively in asymptotic cases, the PoA coincides whether weighted strategy is allowed or not.
Then we consider cases, in which $|N|$ is not asymptotic.
That is the agent size is given.
The following results show that weighted strategy outperforms pure strategy in aspect of PoA.

\smallskip

Finally observe that, in the non-asymptotic case, since weighted delegations enable optimal group accuracy (Theorem \ref{thm:algo}) while pure delegations do not (Example \ref{ex:difference}), the price of anarchy in GD games with weighted delegations is trivially higher than in the case of pure delegations.

\begin{corollary}
$\PoA \geq \PoA^\mathit{pure}$
\end{corollary}


\section{Discussion: Weights as Shares of Votes} \label{sec:shares}

In this paper we interpreted the weight assigned to proxies in a probabilistic way. However, such weights could also be interpreted as an actual allocation of voting power to proxies. The two views lead to different insights, and in this section we briefly develop this second perspective.

\subsection{Weights as Shares of Votes} \label{sec:share}

We move now to interpret weights as shares of voting rights apportioned to proxies. Before any delegation happens each agent is endowed with $1$ vote and may decide to transfer shares of this $1$ vote to other agents. So in that initial situation, if $i$ delegates to $j$ then $\D_{ij}$ is the share of $i$'s vote that is transferred to agent $j$. But in general, $i$ may also receive delegations from other agents, then $D_{ij}$ would correspond to the share of $i$'s current voting power that is transferred to $j$.

Under this interpretation weighted profiles describe the extent of direct transfers of voting power between nodes.  Formally, by the assumption that each agent is endowed with $1$ vote, we denote the weight of all agents by a vector $w\in \mathbb{R}^n$ and the initial weight is $w^{(0)}=(1,\dots,1)$.
Agents then delegate according to $\D$ and obtain the delegated weight 
\[
w^{(1)}=w^{(0)}\D
\]
where the weight of agent $i$ is $w^{(1)}_i=\sum_{j\in N}w^{(0)}_j D_{ji}$. This captures the one-step transfer of the initial voting power. But in liquid democracy these transfers are not only direct: by transferring power to agent $k$ that in turn transfers power to agent $i$, an agent $j$ indirectly transfers shares of her power to $i$. So for each positive integer $t$ we can compute the $t$-steps transfers as  $w^{(t+1)}=w^{(t)}\D$.
We say then that the weighted profile $\D$ converges if for some positive integer $k$, $w^{(k+1)}=w^{(k)}$.

The above process describes a Markov chain and we can use established results to obtain sufficient conditions for the convergence of weighted profiles (cf. \cite{10.2307/2285509} and \cite{gantmacher00theory,proskurnikov17tutorial}) such that, if a weighted profile $\D$ is irreducible\footnote{Recall that a Markov chain is irreducible if each state is accessible by each other state.} and aperiodic, then we can obtain a stationary weight apportionment in finite steps.\footnote{Some periodic Markov chains may converge in infinite steps. We focus on convergence in finite steps as this has a natural interpretation in our setting as an iterated process in which shares of voting power are delegated until a stationary weight distribution (a fixpoint) is reached.}

If $\D$ converges, and we can therefore obtain the `stationary' transfers of voting power within the group, we can also attach a utility for each agent $i$ given $\D$, as follows. Consider $\chi_i^{(0)}\in \mathbb{R}^n$, where the $i$-th element is $1$.
That is $\chi_i^{(0)}$ denotes the initial weight apportionment of agent $i$, where the totality of $i$'s vote is allocated to $i$ herself.
Then for any positive integer $t$, $\chi_i^{(t+1)}=\chi_i^{(t)}\D$ denotes the delegation of $i$'s weight at round $t$.
If the Markov chain of $\D$ satisfies the above condition, we have that $\exists k\in \mathbb{Z}_{++}$, $\chi_i^{(k+1)}=\chi_i^{(k)}=\chi^{(s)}_i$, where $\chi^{(s)}_i$ denotes the stationary weight distribution of agent $i$.
In this case, we say that $\D$ converges on $i$.
Then in $\chi_i^{(s)}$, each positive element is a guru of $i$ with the corresponding weight.
Let then $\chi_i^{(s)}(j)$ denote the share of $i$'s initial vote that is ultimately apportioned to $j$, namely the $j$-th element in $\chi_i^{(s)}$, for any $j\in N$. 

\begin{remark}
It is worth noticing that this notion of agents' weight trivially coincides with the probabilistic one developed in the earlier sections when delegations are only one-hop, that is, when delegation chains are of length at most one. In particular this makes Algorithm \ref{algo:maxMC} relevant for both models.  
\end{remark}

Now that we have a way to compute the weight of agents we can define their utility as a mapping $\hat{U}: \mathbb{D}\rightarrow (N\rightarrow \mathbb{R})$, such that
\begin{align}
\hat{U}_i(\D)= & \left\{
\begin{array}{ll}
\sum_{j\in N}u_i(j){\chi_i^{(s)}(j)} & \mbox{if $\D$ converges on $i$} \\
0 & \mbox{otherwise}
\end{array}
\right.
\label{eq:u_shares}
\end{align}
Observe that we treat the failure to converge, and therefore the lack of a stationary weight apportionment for an agent, as a situation in which the agent gets $0$ utility as she does not get representation by any guru.
Observe furthermore that $\D$ might also converge on some $i\in N$, even though it does not satisfy  the above condition.
For example in the delegation graph of $\D$, $i$ is not linked to a periodic strongly connected component.

\smallskip

We illustrate this alternative way of interpreting weighted delegations by discussing the game that would result from Example \ref{ex:game} using the utility function $\hat{U}$ of equation \eqref{eq:u_shares}. As follows we call $\hat{U}$-NE the NE under the utility of $\hat{U}$.

\begin{example} \label{ex:example1}
Consider the same $2$-player game of Example \ref{ex:game} but now using equation \eqref{eq:u_shares} to determine the utility function. 
We first consider the strategy of agent $1$, given the weighted strategy of agent $2$ is $D_2=(y_{21}, y_{22})$.
If $y_{21}\in (0,1]$, agent $1$ would vote on her own behalf, otherwise $\D$ is not aperiodic.
On the other hand, if $y_{21}=0$, i.e., agent $2$ is a self-voter, agent $1$ would choose to delegate full weight to agent $2$ since $x_{12}>x_{11}$.
We can obtain symmetric best response strategy for agent $2$ given $D_1=(y_{11}, y_{22})$.
Then both agents' optimal strategies are drawn in Fig.~\ref{fig:ex2}.
The $\hat{U}$-NE are $\begin{pmatrix} 0 & 1\\ 0 & 1\end{pmatrix}$ and $\begin{pmatrix} 1 & 0\\ 1 & 0\end{pmatrix}$.
\end{example}
\begin{figure}
\centering
\includegraphics[scale=0.2]{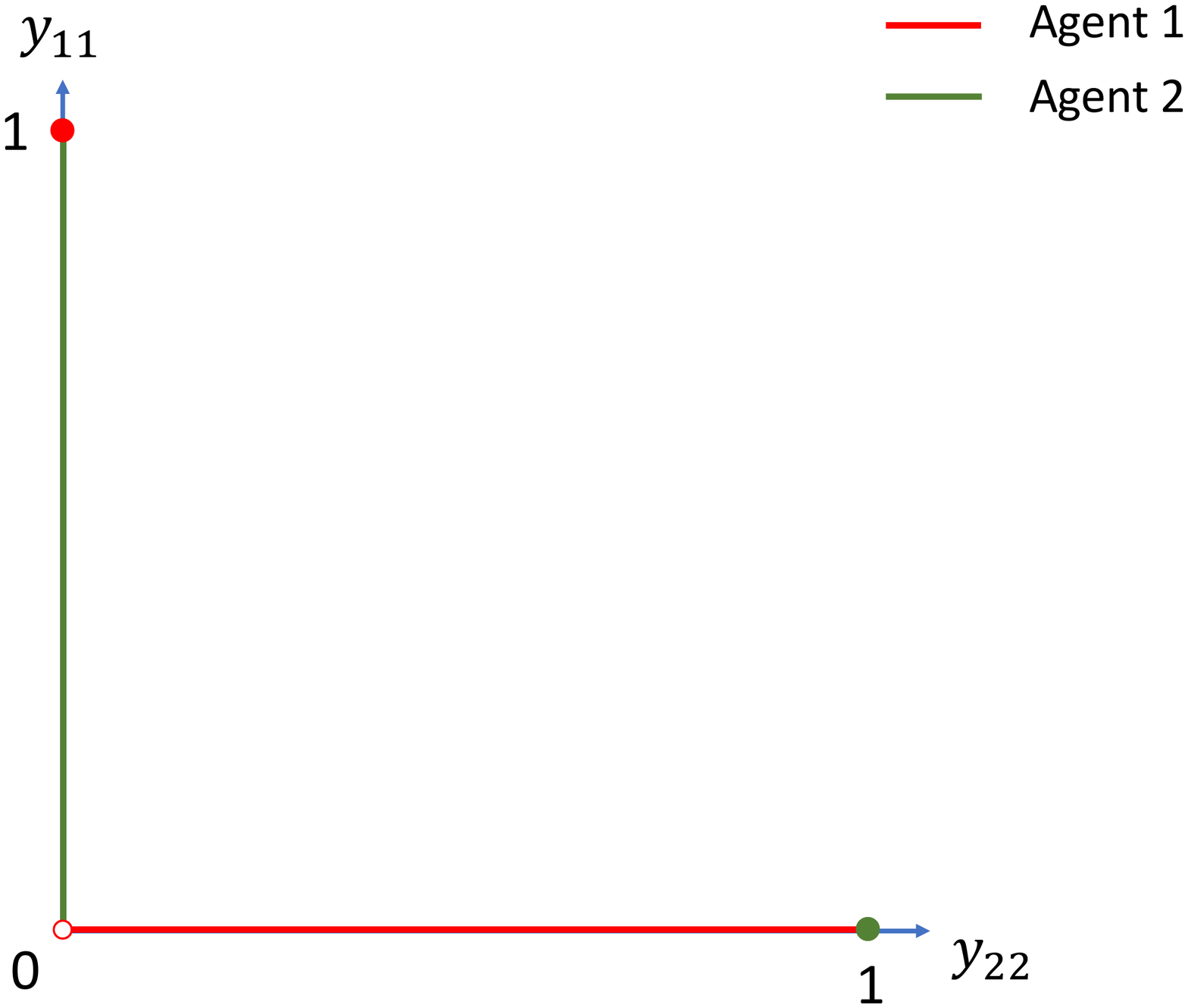}
\caption{Equilibria for the game of Example \ref{ex:example1}}
\label{fig:ex2}
\end{figure}

The example shows that the two notions of utility in equations \eqref{eq:u_prob} and \eqref{eq:u_shares} lead to different types of games, with different equilibria. As we show below, the two notions of utility are independent generalizations of the notion of utility under pure delegations.

\subsection{$U$- and $\hat{U}$-Equilibria in GD Games}

We finally establish a relation between the two types of equilibria in the class of GD games.
Note that the following result holds for general underlying networks.

\begin{theorem} \label{thm:bridge}
For any GD game, if a weighted profile $\D$ is a $U$-NE it is also a $\hat{U}$-NE.
\end{theorem}
\begin{proof}
For any agent $i\in N$, assume the highest-accuracy agent $i$ can access in the underlying network is $i^*\in N$ (there might be multiple highest-accuracy agents and similar argument can be applied), with accuracy of $q^*$.
Since $\D$ is a mixed NE and by Lemma~\ref{lemma:acyclic}, we have that for any $i$ and $j$, if $D_{ij}>0$, any path containing edge $(i,j)$ terminates at $i^*$.
Therefore we have that $\D$ converges and $i$'s full weight is delegated to $i^*$ in $\lim_{k\rightarrow \infty}w^{(k)}_i$.
That is, $i$ cannot further improve her utility, which implies $\D$ is a $\hat{U}$-NE.
\end{proof}
Note that the above theorem does not hold in the other direction. Two-player games offer a simple example. Infinitely many $\hat{U}$-NE exist where the most accurate agent is a sink in $\D$. However, the only $U$-NE in this case is the profile where the maximally accurate agent is a sink and the other agent delegates to her with probability $1$.


\section{Conclusions and Future Work} \label{sec:conclusions}

The paper has proposed a variant of liquid democracy where proxies are weighted by a probability of delegation or, alternative, by a share of the delegator's voting weight. We showed that this type of weighted delegations enable, through central coordination, optimal group accuracy and better equilibria. 

\smallskip

The work presented relies on one strong assumption, namely that agents can delegate to every other agent, so that the interaction network is complete. It will be a priority to lift this assumption in future work. In particular, it has been shown \cite{caragiannis19contribution} that, in the general case in which delegations are constrained by an underlying network, it is intractable to approximate the optimal group accuracy $q_\d$ that could be obtained by centrally determining a standard delegation graph $\d$. Adapting our Algorithm \ref{algo:maxMC} to such general setting would provide an interesting angle from which to approach this optimal delegation problem.

Then, while we studied weighted delegations in a majority voting context from a wisdom of the crowd perspective, it would be interesting to test them in settings involving different types of voting mechanisms. An example are random dictatorships of the sort used, for instance, in proof-of-stake blockchain protocols where stakes can be delegated to proxies in the network (e.g., Tezos \cite{allombert19introduction}). 



\bibliographystyle{plain}



\begin{contact}
Yuzhe Zhang\\
University of Groningen\\
Groningen, the Netherlands\\
\email{yoezy.zhang@rug.nl}
\end{contact}

\begin{contact}
Davide Grossi\\
University of Groningen \\
Groningen, the Netherlands\\
University of Amsterdam \\
Amsterdam, the Netherlands \\
\email{d.grossi@rug.nl}
\end{contact}


\end{document}